\documentclass[preprint,12pt]{elsarticle}

\usepackage{amssymb}
\usepackage{amsmath}
\usepackage{multicol}
\usepackage[ruled]{algorithm2e}
\usepackage{bbding}
\usepackage{algpseudocode}
\usepackage{hyperref}
\newtheorem{theorem}{Theorem}
\newtheorem{lemma}{Lemma}

\newtheorem{propo}{Proposition}
\newdefinition{definition}{Definition}
\newproof{proof}{Proof}

\journal{Journal of Computational and Applied Mathematics}

\begin{document}

\begin{frontmatter}

\title{Proving  Linearizability of Concurrent Stacks}

\author[1]{Tangliu Wen} 
\author[1]{Jie Peng\corref{cor1}}

\address[1]{School of  Information Engineering, Gannan University of Science and Technology,  Ganan 431000, China }
\cortext[cor1]{Corresponding author}
\begin{abstract}
Proving linearizability of concurrent data structures is crucial for ensuring their correctness, but is challenging especially for implementations that employ sophisticated synchronization techniques. In this paper, we propose a new proof technique for verifying linearizability of concurrent stacks. We first prove the soundness of the elimination mechanism, a common optimization used in concurrent stacks, which enables simplifying the linearizability proofs. We then present a stack theorem that reduces the problem of proving linearizability to establishing a set of conditions based on the happened-before order of operations. The key idea is to use an extended partial order to capture when a pop operation can observe the effect of a push operation. We apply our proof technique to verify two concurrent stack algorithms: the Treiber stack and the Time-Stamped stack, demonstrating its practicality. Our approach provides a systematic and compositional way to prove linearizability of concurrent stacks.
\end{abstract}


\end{frontmatter}



\section{Introduction}\label{sec1}
Linearizability \cite{herlihy1} is  a commonly accepted correctness criterion  for concurrent data structures.
Intuitively, linearizability requires that
 each concurrent execution of a concurrent data structure is equivalent
to a legal sequential execution,  and  the sequential execution preserves the order of non-overlapping operations.
Linearizability  imposes  synchronization between the concurrent operations and
affects  performance and scalability of concurrent data structures.
To achieve high performance, concurrent data structures often employ sophisticated fine-grained synchronization techniques \cite{4peng,5haas,6yang,7heller,8hoffman}. This makes it more difficult to prove  linearizability of concurrent data structures.

The standard way to prove linearizability is based on forward or backward simulations.
The proofs involve constructing an invariant which relates the state of the implementation to the state of the specification, and are conceptually
more complex and less amenable to automation.
Henzinger et al. \cite{9h} propose  an aspect-oriented approach to the linearizability proofs  of concurrent queues, which reduces the problem of proving linearizability of concurrent queues to checking four simple properties, each of which can be proved independently by simpler arguments. A key  property is that for two non-overlapping enqueue operations, the value enqueued by the older enqueue operation should be dequeued first.
The proof technique does not rely on program logics, does not need to identify linearization points, and is simple, easy-to-use. They apply the  technique to verify the Herlihy and Wing queue.
However, it is difficult to extend the proof approach to
concurrent stacks.  A challenge for this is that concurrent stacks have no similar properties. For instance, for two non-overlapping push operations, either
of the two values pushed by the two push operations is possible to be popped
first.

In this paper, we extend the Aspect-oriented proof approach  to concurrent stacks.
We first prove the soundness of
elimination mechanism used in the concurrent stacks.
This enables verifiers to  ignore  the elimination pairs and focus on  the ``common'' operations when they prove  linearizability of concurrent stacks.
Then we present a stack theorem,  which reduces the problem of proving linearizability of concurrent stacks to establishing a set of conditions based on the happened-before order of operations.
The basic idea behind the stack theorem is that for any linearizable execution of a concurrent stack, every
pop operation must pop the value pushed by a latest push operation of the
ones whose effects can be observed by the pop operation.
A pop operation will observe the effects of the push operations which precede it.
However, if a pop operation  is
interleaved with  a push operation, the pop operation may observe the effect of the push operation or not.
To solve  the  nondeterministic problem, our stack theorem  employs an extended partial order of the happened-before order such that
if a push operation is not bigger than a pop operation w.r.t. the extended partial order, then  the effect of the push operation can be observed by the pop operation.

Our verification conditions intuitively
express the ``LIFO'' semantics of concurrent stacks and can be verified by just reasoning about the properties based on the happened-before  order of operations, and  do not need to know other proof techniques.
We have successfully applied the proof technique to the Treiber stack \cite{Treiber} and the Time-Stamped (TS) stack \cite{5haas}.

To summarize, our contributions are:
\begin{enumerate}
\item We  prove the soundness of elimination mechanism used in
the concurrent stacks.
\item We present a  simple and complete  proof technique for the linearizability verification of concurrent stacks.
\item We apply our proof technique to prove  two  classic concurrent stacks.
\end{enumerate}

\section{Preliminaries}
\subsection{Linearizability}
In this section, we  introduce basic notations and review the definition of linearizability \cite{herlihy1}.
  An operation is a successful execution of a method.
The calling of a
method $ m$ with  argument $ v$ is represented as an atomic invocation event $ inv_om(v)$, and the return of a method with a return value $v$ is represented  as an atomic response event $ret_o(v)$, where $o$ is an operation identifier.
A thread executing a method  starts with the invocation event, executes its internal atomic events until the  final response event.
We denote an execution of a concurrent threads as a finite sequence of totally ordered atomic events.
A history  is a sequence of invocation and response  events  generated in an execution.
An invocation event matches a response event  if they belong to the same operation.
A history  is sequential if every invocation  event, except possibly the last, is immediately followed by its matching response event.
A sequential history of a concurrent data structure is legal if it satisfies the sequential specification of  the concurrent  data structure.  In the  sequential specification,  all operations are executed atomically.
A history is complete if every invocation event has a matching response event. An invocation event is pending in a history if there is no matching response event to it.
A  completion of an incomplete history $ H$ is a complete history gained by adding some matching response events to the end of $ H$ and removing some pending invocation events within $ H$. Let $\textit Compl(H)$ be the set of all completions of the history $ H$.
Let $ \prec_H$  denote  the  happened-before order of operations in a history $ H$. For any two operations $\textit op_1$ and $\textit op_2$ in $ H$,   we say that $ op_1$ precedes $ op_2$, denoted  $ op_1\prec_H op_2$,  if the response event of $ op_1$ precedes the invocation event of $ op_2$; we say that $ op_1$ is interleaved with $ op_2$,   denoted  $ op_1\simeq_H op_2$, if $ op_1\nprec_H op_2$ and  $ op_2\nprec_H op_1$.

A history $ H$ is linearizable with respect to a sequential specification if there exists a complete history $ C\in \textit Compl(H)$ and a legal sequential history $ S$ such that (1) $ S$  is
a permutation of $C$ ; (2) for any two operations $op_1, op_2$, if $ op_1 \prec_C op_2$, then $ op_1\prec_S op_2$.
$ S$ is called a linearization of $ H$.
A concurrent data structure is linearizable with respect to its sequential specification if every history of the concurrent data structure is linearizable with respect to the sequential specification.

Linearizability is equivalent to observational refinement \cite{filipovic3}. In other words, for a linearizable data structure $Z$ w.r.t. its corresponding  sequential specification $S$, every observable behaviour of any client program using $Z$ can also be observed when the program uses $S$ instead.

Generally, the standard  sequential  ``LIFO'' stack is used to characterize the sequential specification of   concurrent stacks. In this paper, all concurrent stacks we have verified are linearizable with respect  to the standard sequential specification, and we sometimes omit the sequential specification for simplicity.
In this paper, $ push(x)$  denotes a push operation
with an input parameter $x$;  $ pop(x)$ denotes a pop operation with a
return value $ x$.

In this paper, we only consider complete histories. As Henzinger et al.  have shown \cite{9h}, a purely-blocking data
structure is linearizable if every complete history of the concurrent data structure is linearizable. Purely
blocking is a very weak liveness property, and most  concurrent data structures  satisfy the liveness property. All concurrent stacks verified in this paper are purely blocking.

\subsection{Partially Ordered Sets}

The happened-before order of the operations in a history is a strict partial order.
Let $\prec_S$ be a strict partial order  of  the  set $ S$;
 we say  that $y$ is bigger than $x$  if $x,y\in S$ and $ x\prec_S y $; $ x$  is a maximal element in $ S$   if $x\in S$ and $ \forall y\in S.$  $ x\nprec_S y $;  $ x$ is a  minimal  element in $ S$  if  $x\in S$ and $  \forall y\in S.\; y\nprec_S x $.
The partial order $\prec_S$ is  a linear order   if $ \forall x, y\in S. \; x\prec_S y \lor y\prec_S x $.
 Let $ \prec_1$ and $ \prec_2$ be two partial orders of the same set $ S$; the partial order $ \prec_2$ is called an  extension of the partial order $ \prec_1$   if, $\forall a\in S\wedge \forall b\in S \wedge a\prec_1 b$, then $ a\prec_2 b$.
 If  a linear  order (i.e., a total order) is an  extension of a partial order, then it is called a linear  extension of the partial order.

\begin{propo}
Let $ \prec_S$ be a strict partial order on the set $S$, assume the sequence $ L_1, L_2,\dots,L_n $ preserves the partial order   (i.e.,  $\forall i,x,y. \;1\leqslant i,x,y \leqslant n \wedge L_i\in S \wedge  x<y\Longrightarrow L_y\nprec_S L_x$  ).
Then for any element $L'\in S $, $L'$ can be inserted into the sequence such that the new sequence still preserves the partial order $\prec_S$.
\end{propo}

In Appendix \ref{A2}, we show two algorithms by which $L'$  can be inserted into a proper  position such that the new sequence  preserves the partial order.
These algorithms  are used in the proof of  Lemma \ref{elimiLemma} and  Theorem \ref{stackTheorem}.

\section{Soundness of  Elimination Mechanism}
Elimination is a  parallelization optimization technique
for concurrent stacks \cite{HSY,Elimi1}.
A push operation adds an element to the top of a stack and a pop operation removes an element from the top of a stack. Thus, if a push
operation followed by a pop operation is performed on a
stack, the stack's state remains unchanged.
The elimination mechanism is based on the fact.
In a concurrent execution, if a push operation is interleaved with a pop operation and the pop operation pops the value pushed by the push operation, then they  are called an elimination pair.
In order to reduce frequency of shared-data accesses and  increase
the degree of parallelism, the elimination mechanism allows the elimination pairs to exchange their
values  without accessing the shared stack structure.

The different implementations of elimination mechanism in  different stacks \cite{HSY,Elimi1,TSStack} have been proven that it does not violate linearizability of concurrent stacks.
Next, we prove soundness of the elimination mechanism
 in  general case (i.e., not considering concrete stacks).

Let $\textit  Push(H)$ and $\textit  Pop(H)$  denote the sets of all push  and pop operations  in  a history $ H$, respectively. For simplicity, we assume
that all values  pushed by push operations are unique. We formalize elimination pairs as follows:

\begin{definition}
A pop operation $pop$ and a push operation $push$  is  an elimination pair in $ H$ if $push\in \textit Push(H)\,\land\,pop\in \textit Pop(H)\,\land\, push\simeq_H pop\,   \land\, \textit Mat(pop)=push$.
\end{definition}

The function $\textit Mat$ maps each pop operation to the push operation whose value is popped by the pop operation, will be defined formally in the next section.

\begin{lemma}\label{elimiLemma}
For a history  $H$ of a concurrent stack, let $H_s$ be a subsequence of $H$ obtained by deleting an elimination pair of $H$. If $H_s$ is linearizable with respect to the standard sequential   stack specification, then $H$ is also linearizable with respect to  the specification.
 \end{lemma}

\begin{proof}
Assume that $ push(x)$ and $ pop(x)$ are an elimination pair in $ H$.
 $ H_s$ is obtained from $ H$ by deleting the  elimination pair.
Assume that $ H_s'$ is a linearization of $ H_s$.
Since $ H_s'$ is a sequential history,
for any operation in $ H_s'$, we directly use the operation to replace its  invocation and response events. Consider the following two cases:
 in the first case, $ push(x)$ begins to execute  earlier than $ pop(x)$, as shown in Fig. 1; in the second case,  $ pop(x)$ begins to execute  earlier than  $ push(x) $.
\begin{figure}[h]
  \centering
  \includegraphics[height=0.12\textwidth]{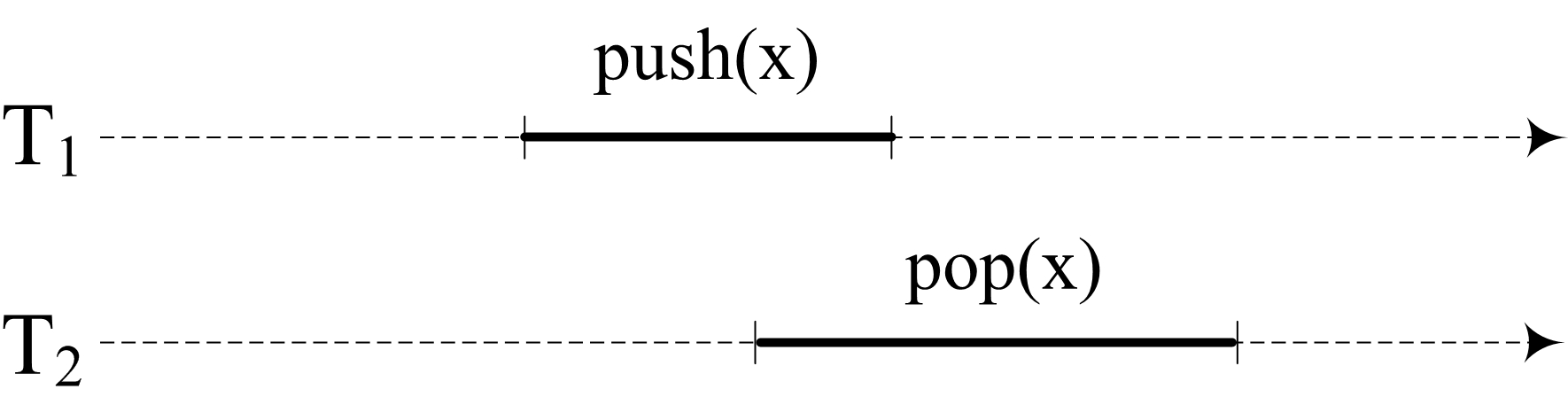}
  \caption{$ push(x)$ begins to execute  earlier than $ pop(x)$ }
 \end{figure}

In the following, we prove that  the lemma holds  in the first case.
The proof is done in two stages. Firstly,
we construct a new sequence   $ H'$  by the following steps:
we  insert $ pop(x)$ into $ H_s'$ by using Algorithm 1; then, insert $ push(x)$  just before $ pop(x)$.
Secondly, we prove that the new sequence  $ H'$ (strictly speaking, $ H'$ corresponds to the sequential history) is  a linearization  of $ H$.
The  proof for the second case is a similar process.
In the first case, we can easily get the following properties according to Fig. 1:

\begin{center}
Property 1.  $\forall op'.\; pop(x)\prec_H op'\Longrightarrow push(x) \prec_H op'$.
\\
 Property 2.  $\forall op'.\; op'\simeq_H pop(x)\Longrightarrow  op'\nprec_H push(x)$.
 \\
Property 3.  $\forall op'.\; op'\prec_H pop(x)\Longrightarrow push(x)\nprec_H  op'$.
\end{center}

We consider the following  two cases:

\textbf{Case 1.}  In $ H_s'$,  there is  no any  operation which precedes $ pop(x) $. In this case, after the two inserting operations, the sequence $ H'$ is the following form:
\[H'=push(x)^\smallfrown pop(x)^\smallfrown H_s'\]
By  Algorithm 1, the sequence $ pop(x)^\smallfrown H_s'$ preserves the happened-before order. By Property 1 and  Property 2,  any operation in  $ H_s'$ does not precede $ push(x)$. Thus, the new sequence $ H'$ preserves the happened-before order.
Obviously,  $ H'$ satisfies  the ``LIFO'' semantics. Thus, $ H'$  is a linearization of $ H$.

\textbf{Case 2.} In $ H_s'$, there exist  operations which precede $ pop(x)$.
In this case, after the step of  inserting $  pop(x) $,  the  operation just before $  pop(x) $  precedes $  pop(x) $ (by Algorithm 1). Assume that the  operation  just  before $  pop(x) $ is $op$.  After  the step of  inserting $  push(x) $ just before $ pop(x)$, the sequence is the following form:
\[H'=\dots,op^\smallfrown push(x)^\smallfrown pop(x),\dots\]
Since $ op\prec_H  pop(x)$, $ op$  either be  interleaved with $ push(x)$ or precede $ push(x)$. The   former  is  shown in Fig. 2. The operation $op$  finishes before $pop(x)$ begins and it is interleaved with $push(x)$.
In either case, there are the following properties:

\begin{center}
 Property 4.
 $\forall op'.\; op'\prec_H op\Longrightarrow   push(x)\nprec_H op'$.
 \\
 Property 5.
 $\forall op'.\; op'\simeq_H op\Longrightarrow  push(x)\nprec_H op' $.
\end{center}

By Property 1 and Property 2,  any operation on the right of  $ pop(x)$ does not precede $ push(x)$.
By Property 3, $ push(x)$ does not  precede $ op$.
By Property 4 and Property 5, $ push(x)$ does not  precede  any operation on the left of  $ op$.
Thus, the new sequence  $ H'$ preserves the happened-before order.
Obviously,  $ H'$ satisfies  the ``LIFO'' semantics.
Thus,  $ H'$ is a linearization of $ H$.
\end{proof}

\begin{figure}[ht]
 \centering
\includegraphics[width=2.2in]{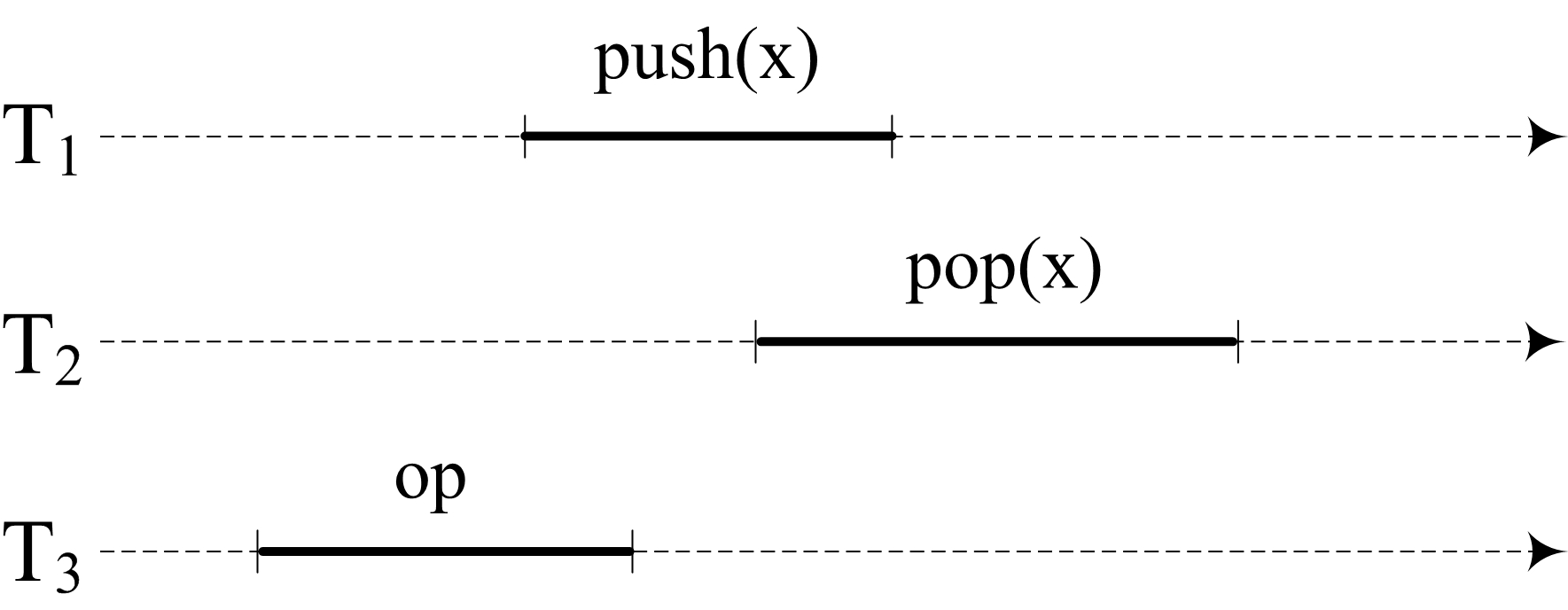}
\caption{the operation $ op$  is interleaved with $ push(x)$}
 \end{figure}

By the above lemma, we can get the following theorem.

\begin{theorem}\label{elimiTheorem}
For a history  $H$ of a concurrent stack, let $H_s$ be a subsequence of $H$  obtained  by deleting all elimination pairs of $H$. If $H_s$ is linearizable with respect to  the standard
sequential  stack specification, then $H$ is also linearizable with respect to  the specification.
 \end{theorem}

\section{Conditions for Linearizability of Concurrent Stacks}
This section first describes our motivation and challenges of verifying concurrent stacks, then presents the stack theorem which reduces the problem of proving linearizability to establishing a set of conditions based on the happened-before order of operations, and proves that they are necessary and sufficient conditions for proving linearizability of stacks.
\subsection{Motivation and Challenges of  Concurrent Stacks}
Henzinger et al. \cite{9h} propose  an aspect-oriented approach to the linearizability proofs  of concurrent queues, which reduces the problem of proving linearizability of concurrent queues to checking four simple properties.
A key property is  that for any two non-overlapping enqueue
operations $enq_1$ and $enq_2$ in a concurrent execution, if $enq_1$ precedes $enq_2$ (w.r.t. the happened-before order), then the value inserted by $enq_2$ cannot be
removed earlier than the one inserted by $enq_1$, i.e., $deq_2$ cannot precede $deq_1$, where $deq_2/deq_1$ removes the value
inserted by $enq_2/enq_1$.
In a concurrent queue, as enqueue and dequeue operations contend for the tail and head respectively. However, in a concurrent stack, the top pointer is typically the unique hot spot contended by pop and push operations. It is difficult  to extend the  proof approach  to concurrent stacks. A challenge for this is that concurrent stacks have no similar properties.  For instance, for two non-overlapping push operations,  either  of the two values pushed by the two push operations is possible to be popped first.

Our basic idea is that in a linearizable execution of a concurrent stack, every pop operation must pop the value pushed by a latest  push operation  of the ones whose effects can be observed by the pop operation.
For example,  in the execution shown in Fig. 3, $push(x)$, $push(y)$ and $push(z)$ precede $pop(z)$, and $pop(z)$ can observes their effects.
$pop(z)$ pops the value pushed by  $push(z)$ which is a latest one (i.e., a minimal one w.r.t. the happened-before order)  of the three  push operations.
After the value $z$ is popped,  $pop(y)$ can observe the effects of $push(x) $ and $push(y)$ and  pops the value $y$ pushed by  $push(y)$ which is the latest one   of the two  push operations.
If $pop(z)$ or $pop(y)$ pops the value pushed by  $push(x)$, which is not the latest one, the execution
will not be linearizable.

\begin{figure}[ht]
  \centering
  \includegraphics[height=0.23\textwidth]{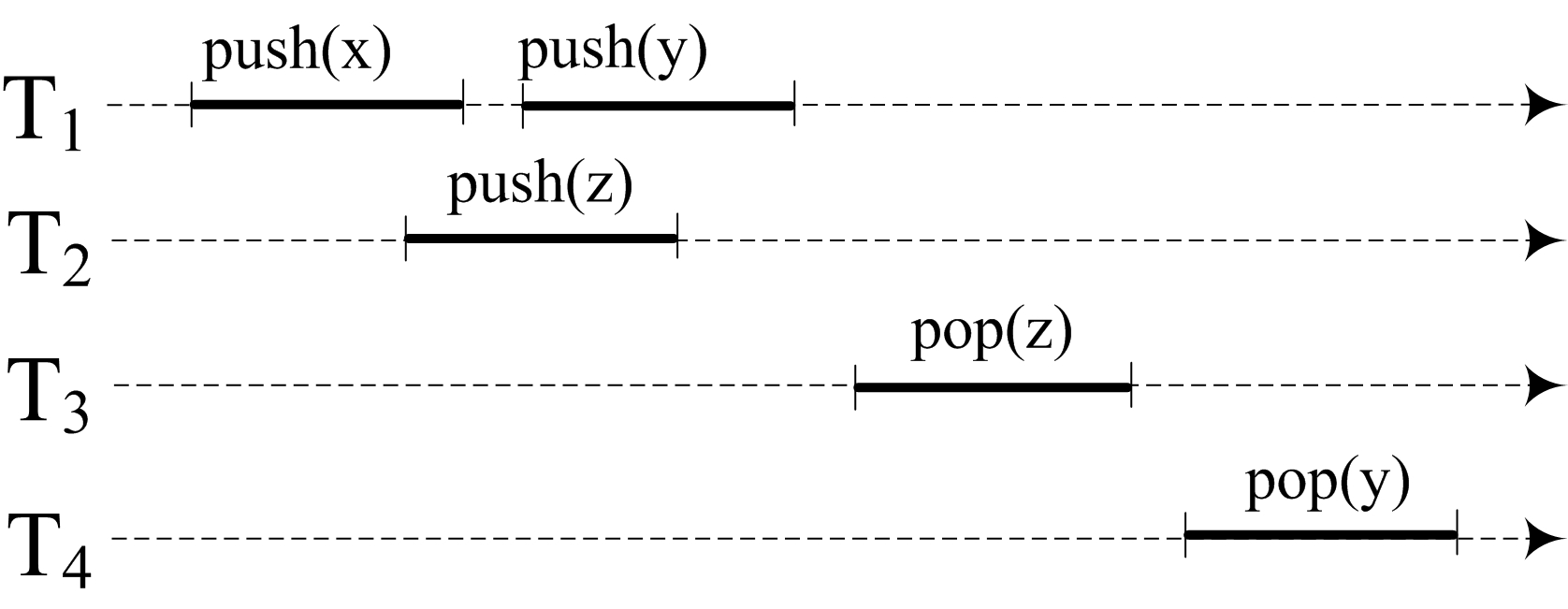}
  \caption{Example execution}
 \end{figure}

A pop operation will observe the effects of the push operations which precede it.
However, if a pop operation  is
interleaved with  a push operation, the pop operation may observe the effect of the push operation or not.
For example, in the execution shown in Fig. 4, $pop(o)$ is interleaved with $push(m)$.
$pop(o)$ does not observe the effect of $push(m)$, otherwise, $pop(o) $ should  pop the value $ m $ of $push(m)$ (because $push(o)$ precedes $push(m)$).

\begin{figure}[ht]
  \centering
  \includegraphics[height=0.23\textwidth]{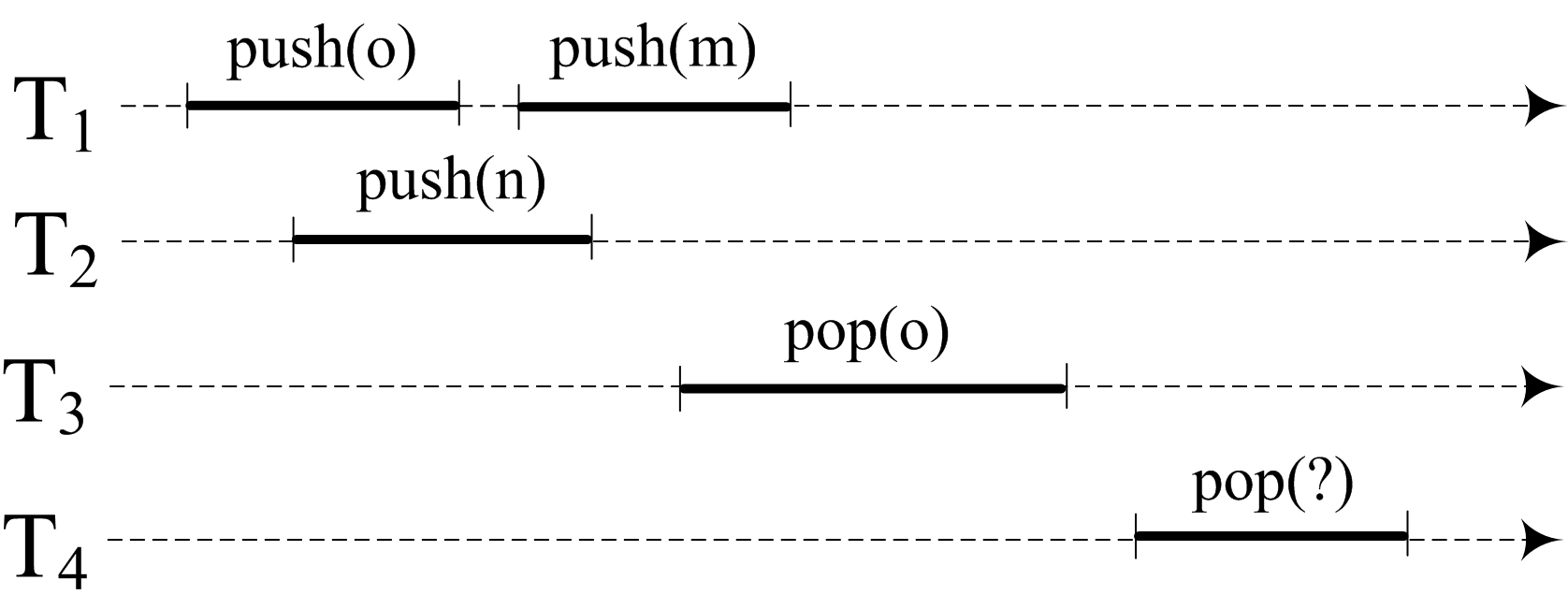}
  \caption{$ pop(?)$ cannot pop the value $n$ }
 \end{figure}

To solve  the above nondeterministic problem, our stack theorem  employs an extended partial order of the happened-before order such that
if a push operation is not bigger than a pop operation w.r.t. the extended partial order, then  the effect of the push operation can be observed by the pop operation.
Informally, in an execution of a concurrent stack, if every pop operation  pops the value pushed by a maximal  push operation  of the ones which are not bigger than the pop operation (w.r.t. the extended partial order), then the execution is linearizable.

\subsection{Stack Theorem}

For simplicity, we assume
that all values  pushed by push operations are unique.
We map each pop operation to the push operation whose value is popped by the pop operation, or to $ \epsilon$ if the pop operation returns $\textit empty$. We say that a mapping is safe if (1) a pop operation always pops the value pushed by a push operation or returns $\textit empty$; (2) a value is popped by a pop operation at most once.
Obviously, every  linearizable history of a concurrent stack has  a unique safe matching.
If there does not exist a safe matching, then the history is not linearizable.
We formalize the notion as follows:

\begin{definition}
A mapping $\textit Mat$  from $\textit Pop(H)$  to $\textit Push(H)\cup\{\epsilon\}$  is safe if
\begin{enumerate}
 \item  $  \forall  pop\in\textit Pop(H)$,  if $\textit Mat(pop) \neq\epsilon$, then the value popped by the pop operation  is pushed by the push operation $ \textit Mat(pop)$.
 \item  $ \forall pop\in \textit Pop(H)$, if $\textit Mat(pop)= \epsilon$, then the  pop operation returns $\textit empty$.
 \item  $ \forall pop, pop'\in\textit Pop(H).$ if $ pop\neq pop' \wedge\textit Mat(pop)\neq\epsilon$, then  $\textit Mat(pop)\neq \textit Mat(pop')$.
 \end{enumerate}
\end{definition}

For a history $ H$, let $ \textit Mat^{\textit{-1}}$ map the set of  push operations whose values have been popped  to the set of non-empty pop operations such that $ \forall pop\in \textit Pop(H)$, $  \forall push\in \textit Push(H). $  $\textit  Mat(pop)=push\Longrightarrow \textit Mat^{\textit{-1}}(push)=pop $.
The mapping $\textit Mat^{\textit{-1}}$ is used in the proof of Theorem \ref{stackTheorem}.

We say that the sequence $pop_1,\dots,pop_n$
is a linearization of pop operations of $ H$, if
$\textit Pop(H)\!=\!\{pop_1, \dots, pop_n\}\wedge x\!<\!y \Longrightarrow pop_y\nprec_H pop_x $.
Similarly, we say that the sequence $push_1,\dots,push_m$
is a linearization of push operations of $H$, if
$\textit Push(H)\!=\!\{push_1, \dots, push_m\}\wedge x\!<\!y \Longrightarrow push_y\nprec_H push_x $.

We now state the stack theorem,
which gives necessary and sufficient conditions for linearizability of concurrent stacks.
Any  history comprising only the events of push operations is always linearizable.
The reason is that linearizability is a property of externally observable behaviors (i.e., histories) and  the return value of a push operation is always  $\textit null$ or the signal $\textit ok $.
Such histories can be ignored  when  we verify linearizability of concurrent stacks. Our stack theorem is stated below.

\begin{theorem}\label{stackTheorem}
Let $H$ be a complete history of a concurrent stack  containing the events of pop operations.  $H$  is linearizable with respect to  the standard  sequential  stack specification  iff
there exists a linearization $pop_1, pop_2, \dots, pop_n$ of  pop operations,   an extension
 $\prec_e $ of the happened-before order $\prec_H$   and a safe mapping $\textit  Mat$ from  $\textit  Pop(H)$  to $\textit  Push(H)\cup\{\epsilon\}$,  such that
 \begin{enumerate}
 \item  $\forall i,j.\; 1<j\leqslant n$,  $\textit    1\leqslant i<j. \; \textit  Mat(pop_i)\neq \epsilon \Longrightarrow pop_j \nprec_H \textit  Mat(pop_i)$;
 \item   $1\leqslant\forall i\leqslant n$. if $\textit  Mat(pop_i)\neq\epsilon$,    then  $\textit Mat(pop_i)\in B_i $ and  $ \forall push $ $\in B_i.\;\textit  Mat(pop_i)\nprec_e push$;
\item  $1\leqslant \forall i\leqslant n$. if  $\textit  Mat(pop_i)=\epsilon$,  then (a)
 $ \forall push\in \textit  Push(H).\; push\prec_H pop_i \Longrightarrow push \in A_i$; (b) $\forall push\in C_i ,\forall x\leqslant i-1. \; push \nprec_H pop_x $ and $ push\nprec_H \textit Mat(pop_x)$.
 \end{enumerate}
\end{theorem}

\noindent
Here, $1\leqslant\forall i\leqslant n$,  the set $A_i=\{\textit Mat(pop_1),\dots,\textit Mat(pop_{i-1})\}$, $B_i=\{push\mid push\in \textit  Push(H)- A_i \wedge \; pop_i\nprec_e  push   \}$, $ C_i=\{push\mid push\simeq_H pop_i \wedge  push\in \textit  Push(H)-A_i \}$.

The set $ A_i$ is a set of the push operations whose values  have  been popped by the pop operations ahead of $ pop_i$.
 $ B_i $ is a set of the push operations whose values have not been popped by the pop operations ahead of $ pop_i$ and which is not bigger than $ pop_i$  w.r.t.  $ \prec_e$. This means that the effects of the push operations in the set $ B_i $  can be observed by $ pop_i$.
 Informally,
the second condition requires that every non-empty pop operation always pops the value pushed by a latest  push operation (i.e., maximal w.r.t. $ \prec_e$) in the set $ B_i$.

$ C_i$ is a set of the push operations  which are interleaved with $ pop_i$  and whose values have not been popped by the pop operations ahead of $  pop_i$.
The third condition requires that  for any push operation which  precedes the empty pop operation $ pop_i$, the value of the push operation has been popped by a  previous pop operation (i.e., by $ pop_x, x<i$);
for any push operation which is interleaved with the empty pop operation, if its value  is not popped by a previous pop operation,  then the push operation does not precede the previous pop operations and their matched  push operations.

\begin{proof} ($\Longleftarrow$).
We first prove that the theorem holds when $ H$ does not contain empty pop operations, then further extend the result to the case where $ H$ contains  empty pop operations.

\vspace{0.5em}
1. $ H$ is linearizable when $ H$ does not contain
 empty pop operations.

\textbf{Proof.} The proof is done in the following steps. Firstly,
  we  construct a linearization of push operations.
  Secondly, we insert $pop_1,pop_2,\dots,pop_n$, one
after another, into the linearization sequence of  push operations.  Finally, we show  that the final sequence is a linearization of $ H$.

\textbf{Step 1.}   Construct a linearization of push operations.

Assume  the number of push operations in $ Push(H)$ is \textit{m}. We construct a linearization of push operations by the following rules:
First,
if all the maximal  push operations w.r.t. the partial order $  \prec_e$ in $ Push(H)$  have no matching pop operations,
then choose any  one from the maximal  push operations; otherwise,
choose the one from the maximal
push operations  such that  its matching pop operation  is the smallest w.r.t.  the linearization order of  pop operations.
 Formally,  Let $ push_m$ denote the chosen element, let $ S=\{x \mid x\in Push(H)\land \forall  y\in Push(H). \; x\nprec_e y\}$, $  push_m$ $\in S $ and  $ \forall  push'\in S-\{push_m\}$,
if $ push_m$  and  $ push'$ has a matching pop operation, then
$\textit Mat^{\textit{-1}}(push_m) \prec_L \textit Mat^{\textit{-1}}(push')$, where $  \prec_L $ denotes the linearization order of the pop operations.

Second, delete $ push_m$  from $ Push(H)$; By the same rule, choose an element $ push_{m-1}$ from the rest of $ Push(H)$; and insert it before $ push_m$ (i.e., $ push_{m-1}, push_m$).
Finally, continue to construct the sequence  in the same way until all push operations of  $ Push(H)$ are chosen.

 By the above construction, the final sequence   $push_1, \dots,push_m$
 is a linearization of all push operations, and has the following property:

\textbf{Property  1.}  For any push operation $ push_x$, $ push_x$ is a maximal  push operation w.r.t.  $  \prec_e $, among the  push operations on the left of  $ push_x$ (i.e., the set $ \{push_1,\dots,push_x\}$), and if $ push_x$ has a matching pop operation, its matching pop operation is the smallest w.r.t.  $  \prec_L $, among  the matching pop operations of the push operations which are in the set $ \{push_1,\dots,push_x\}$ and are not smaller than $push_x$ w.r.t.  $  \prec_e $.

\textbf{Step 2.} We insert $pop_1,\dots,pop_n$, one
after another, into the sequence $push_1,\dots,push_m$.
Every time we insert a pop operation, we show that the new sequence satisfies the ``LIFO'' semantics and preserves  the happened-before order  $  \prec_H$.

\textbf{Step 2.1.} Using Algorithm 2,  we  insert
$  pop_1$ into the linearization sequence of push operations. After the inserting operation, the new sequence preserves  $  \prec_H$.

If $   \textit  Mat(pop_1)$ is on the left of $  pop_1$, then by  $  Property \;1$   and  the second condition of the theorem,  $ \textit Mat(pop_1)$ is  just before $  pop_1$  (i.e.,  $  \dots, \textit  Mat(pop_1),pop_1,\dots $).
If $ \textit Mat(pop_1)$ is on the right of $   pop_1$,
then $   pop_1$  and $ \textit Mat(pop_1)$ is an elimination pair. The reasons are as follows:
Let $  push_x$ denote the push operation just after $   pop_1$.
By Algorithm 2, $   pop_1 \prec_e push_x$ , thus  $ \textit Mat(pop_1)\neq push_x$.
If $  \textit Mat(pop_1)\prec_H pop_1$, then $  \textit Mat(pop_1) \prec_e pop_1$. By  $   pop_1\prec_e push_x $, then
$ \textit Mat(pop_1) \prec_e push_x$, contradicting the fact  that
$ \textit Mat(pop_1)$ is on the right of $  push_x$.
In this case, we delete the elimination pair from the sequence.
After the step,  the new sequence satisfies the ``LIFO'' semantics and preserves  $ \prec_H$.

\textbf{Step 2.2.}  Similar to $   pop_1$, for any  $   pop_i$, $  2\leqslant i\leqslant n $, insert it into the new sequence.

In the new sequence obtained after inserting $  pop_{i-1}$,  we use   Algorithm 2 to search the first push operation which is bigger than $  pop_i$ w.r.t. $  \prec_e$.
If such push operation  (denoted by $  push_y$) is found,
we insert $  pop_i$ just before $  push_y$; otherwise (in this case, the new sequence has no any push operation which is bigger than
$  pop_i$ w.r.t. $  \prec_e$), we insert $   pop_i$ into the end of the new sequence.

After the step of  inserting $   pop_i$,
if $ \textit Mat(pop_i)$ is on the right of $   pop_i$,
we delete them from the new sequence (in this case, $   pop_i$  and $ \textit Mat(pop_i)$ is an elimination pair);
Otherwise, we consider the following two cases:

\textbf{Case 1.}  $   pop_i$ is not between  any previous pop operation $ pop_x, x<i$ and its matching push operation $\textit Mat(pop_x)$.
In this case, we prove the new sequence satisfies the ``LIFO'' semantics and preserves $  \prec_H $.

Since $ \textit Mat(pop_i)$ is on the left of $   pop_i$ in this case,   by  $  Property \;1$ and the second condition of  the theorem,  we get that the new sequence  satisfies the ``LIFO'' semantics.
Next, we prove that the new sequence preserves $  \prec_H$.
According to Algorithm 2, $  pop_i$ and any push operation of  the new sequence do  not  violate $  \prec_H$.
According to the order of inserting pop operations, $  pop_i$  and any pop operation on the left of $  push_y$  do not violate $  \prec_H$.
$  pop_i$  and any pop operation on the right of $  push_y$  do not violate $  \prec_H$.
The reason for this is as follows.
For any pop operation  $  pop_x$ on the  right of $  push_y$,
assume  by contradiction  that $  pop_x\prec_H pop_i $.
By the assumption, we get  $  pop_x\prec_e pop_i $.
Since $  pop_i\prec_e push_y $, $  pop_x\prec_e push_y$.
Since $  pop_x\prec_e push_y$, $  pop_x$ is not on the  right of $  push_y$ (by Algorithm 2 and the rule of constructing the new sequence),  thereby yielding a contradiction.

\textbf{Case 2.}  If $   pop_i$ is  between   some previous pop operations and their matching push operations, assume that $   pop_z$ is the last one among the pop operations and $ \textit Mat(pop_z)=push_z$, where $z<i$.
then move $   pop_i$ to the right of $   pop_z$, as follows.
\[\boldmath{seq:} \dots, push_z, \dots,  pop_i,push_y,\; \dots,\; pop_z, \dots\]
\[\boldmath{seq':}  \dots, push_z, \dots,  push(y), \dots, \;pop_z, pop_i, \dots\]

$  seq$ and $  seq'$ are the sequences  before and after the move transforming, respectively.
In  the new sequence $  seq'$, $  pop_i$ is not between any previous pop operation and its matching push operation.
Since $ \textit Mat(pop_i)$ is on the left of $   pop_i$ in this case,
by  $  Property\; 1$ and the second  condition of  the theorem,   $  seq'$  satisfies the ``LIFO'' semantics.
Next, we show that $  seq'$ preserves  $  \prec_H $.

Before inserting $   pop_i$,  the sequence satisfies the  ``LIFO'' semantics. Thus, for all push operations between $   push_z$ and $   pop_z$,  their matching pop operations are  between $   push_z$ and $   pop_z$.
By the first condition of the theorem and $z<i$, $   pop_i$  does not precede any push operation  between $   push_z$ and $   pop_z$.
By  the order of inserting pop operations and $z<i$,
$   pop_i$ does not precede any pop operation  between $   push_z$ and $   pop_z$.
Thus, after the move transforming, $   pop_i$  and any operation  between $   push_z$ and $   pop_z$  do  not  violate $  \prec_H$.
Similar to the proof in  Case 1, we can get that  $   pop_i$  and any other operation do  not  violate $  \prec_H$.
Thus,  the new sequence $  seq'$ preserves  $  \prec_H $.

\vspace{0.5em}
2. $  H$ is also linearizable when $  H$ contains
 empty pop operations.

\textbf{Proof.} We construct a linearization $  H'$ of $  H$ by the following process: if $ \textit Mat(pop_i)=\epsilon$, let $  A$ denote the linearization of  $pop_1,\dots,pop_{i-1}$ and their matching push operations (which is constructed by the above method), let $  B$ denote the linearization of the rest of $  H$.
Let $  H'=A^\smallfrown (pop_i)^\smallfrown B$.

By  the process of  constructing $  H'$, $  H'$ satisfies the ``LIFO'' semantics.
Next, we show that $  H'$ preserves  $  \prec_H$. Obviously, any two pop operations do not violate $  \prec_H $ in $  H'$.
In the following, we show that in $  H'$, (1) any two push operations do not violate $  \prec_H $, and
 (2) any pop operation and any push operation do not violate $  \prec_H $.

Let $  push_x$/$  pop_x$ be a push/pop operation in
$ A$.
Let $  push_y$/$  pop_y$ be a push/pop operation in
$ B$.
By the first condition of the theorem, we can get that
$  pop_y\nprec_H push_x$ and $  pop_i \nprec_H push_x$.
By the first part of the  third  condition, we can get $  push_y\nprec_H pop_i$.
If $   pop_i\prec_H push_y$, then $  push_y \nprec_H push_x$ and
$  push_y \nprec_H pop_x$.
If $   push_y \simeq_H pop_i$, then by the second part of the  third  condition, we can get  $  push_y \nprec_H push_x$ and
$  push_y \nprec_H pop_x$.

($\Longrightarrow$)
 \textbf{Proof.} Since $  H$ is linearizable, there exists a safe
 mapping  $\textit{Mat}$  from $  Pop(H)$  to $  Push(H)\cup\{\epsilon\}$.
We assume that $  H'$ is a linearization of  $  H$. The linear order  $  \prec_{H'}$ is an extension of  $  \prec_H$.
Let $pop_1,\dots,pop_n$ be  the maximal subsequence of
$  H'$ consisting of  pop operations. Obviously, it is a linearization of  pop operations.
Based on the extended partial order $  \prec_{H'}$,
 the linearization of  pop operations and the safe
 mapping  $ \textit Mat $,  we show  that the three conditions of Theorem \ref{stackTheorem} hold.

\vspace{0.5em}
1. $  \forall i,j.\; 1<j\leqslant n$,  $   1\leqslant i<j. \;\textit Mat(pop_i)\neq \epsilon \Longrightarrow pop_j \nprec_H\textit Mat(pop_i)$.

\textbf{Proof.}
Since $  pop_i \prec_{H'} pop_j$
and $ \textit Mat(pop_i) \prec_{H'} pop_i$,   $ \textit Mat(pop_i) \prec_{H'} pop_j$. Thus,  $  pop_j \nprec_H\textit Mat(pop_i)$ (because $  \prec_{H'}$ is a linear extension of  $  \prec_{H}$).

\vspace{0.5em}
2. $1\leqslant\forall i\leqslant n$. if $ \textit Mat(pop_i)\neq\epsilon$, let the set $B_i=\{push\mid push\in \textit  Push(H)-\{\textit Mat(pop_1),$ $\ldots,$ $\textit Mat(pop_{i-1})\} $ $ \wedge \; pop_i\nprec_{H'}  push   \}$,   then  $\textit Mat(pop_i)\in B_i $ and  $ \forall push $ $\in B_i.\;  \textit  Mat(pop_i)\nprec_{H'} push$.

 \textbf{Proof.}
 Since $ H'$ is a linearization of  $ H$, $ \textit Mat(pop_i) \prec_{H'} pop_i$. Thus,  $  \textit  Mat(pop_i)\in B_i $.
 Since $   pop_i\nprec_{H'}  push$ (by  $push \in B_i $),  we can get $  push \prec_{H'} pop_i$ (because $  \prec_{H'}$ is a linear order).
Since $  push\in B_i$,
the value pushed by $  push$ is not popped before $  pop_i$ in $  H'$.
Since $  H'$  satisfies the ``LIFO'' semantics, $   push \prec_{H'}  \textit  Mat(pop_i) \prec_{H'}  pop_i$.

\vspace{0.5em}
3. $1\leqslant\forall i\leqslant n$. if  $\textit  Mat(pop_i)=\epsilon$, let $A_i=\{\textit Mat(pop_1),\dots,\textit Mat(pop_{i-1})\}$, and $ C_i=\{push\mid push\simeq_H pop_i \wedge  push\in \textit  Push(H)-A_i \}$, then (a)
 $ \forall push\in \textit  Push(H).\; push\prec_H pop_i \Longrightarrow push \in A_i$; (b) $\forall push\in C_i ,\forall x\leqslant i-1. \; push \nprec_H pop_x $ and $ push\nprec_H \textit Mat(pop_x)$.

 \textbf{Proof.} If a push operation $push\prec_H pop_i $, then $   push\prec_{H'} deq_i $.  Since $  H'$  satisfies the ``LIFO'' semantics, the value pushed by $  push$ is popped by a previous pop operation, i.e., $ push \in A_i$.
If a push operation $  push\in C_i$, then $   pop_i \prec_{H'} push $.
Since $ \forall x\leqslant i-1, \; pop_x \prec_{H'} pop_i$ and $ \;Mat(pop_x)  \prec_{H'} pop_i$,  we get $  pop_x \prec_{H'} push $ and $ \; \textit Mat(pop_x) \prec_{H'} push $.
Thus, $  push \nprec_H pop_x$ and $  push\nprec_H  \textit  Mat(pop_x)$.
 \end{proof}

Consider the  execution  shown in Fig. 4. If $pop(?)$ pops the value pushed by $push(m)$, then the execution is linearizable. In this case, let $pop(m)$ denote the pop operation $pop(?)$.
Next, we will  prove linearizability of the execution using the stack theorem.

The history has a unique safe matching, where $\textit Mat(pop(o))=push(o)$ and   $\textit Mat(pop(m))=push(m) $.
We  construct an extended partial order $\prec_e $ of  $\prec_H$ such that $pop(o)\prec_e push(m)$.
Since  $push(n)\prec_H pop(o)$ and $pop(o)\prec_e push(m)$, $push(n)\prec_e pop(o)\prec_e push(m)$.
We choose the only linearization of pop operations, $pop(o), pop(m)$,  to verify linearizability of the execution. We add subscripts for the two pop operations, i.e.,  $ pop_1(o),  pop_2(m) $.

Since  $pop_2(m) \nprec_h \textit Mat(pop_1(o))=push(o)$, the first condition is satisfied.
For the pop operation, $pop_1(o)$,  the set $B_1=\{push(o),push(n)\}$, is a set of the push operations  which is not bigger than $ pop_1(o)$  w.r.t.  $ \prec_e$. $ push(o)$ is maximal w.r.t.  $ \prec_e$ in the set $B_1$.  Since  $\textit Mat(pop(o))=push(o)$,  the second condition is satisfied in this case.
For the pop operation, $pop_2(m)$,   the set $B_2=\{push(m), push(n)\}$, is a set of the push operations whose values have not been popped by the pop operations ahead of $pop_2(m)$ and which is not bigger than $pop_2(m)$  w.r.t.  $ \prec_e$.
$push(m)$ is maximal in the set $B_2$ w.r.t.  $ \prec_e$.  Since $\textit  Mat(pop(m))=push(m)$,  the second condition is satisfied in this case.

Thus, if $pop(?)$ pops the value pushed by $push(m)$, then the execution is linearizable.
 Note that in the execution, if $pop_2(?)$ pops the value of $push(n)$, then the execution will not  be linearizable. This is because there does not exist an extended partial order such that both of  the two pop operations satisfy the second condition.

In Section 5.2, we  will show how to construct  a  linearization of a specific execution  using the above method used in the proof of Theorem \ref{stackTheorem}.

\section{Verifying the Treiber Stack and  the TS stack}

In this section, we first illustrate the proof technique on a simple concurrent stack: the Treiber stack \cite{Treiber}, then a sophisticated concurrent stack: the Time-Stamped (TS) stack \cite{TSStack}.

\subsection{Verifying the Treiber Stack}

The Treiber stack (Fig. 5) is a lock-free concurrent stack  based on a singly-linked list with a top pointer $S$.   The push and
pop methods, which avoid the overhead of acquiring and releasing locks,  try to update the top pointer using CAS instructions to finish their operations.

\begin{figure}[ht]
  \centering
  \includegraphics[height=0.25\textwidth]{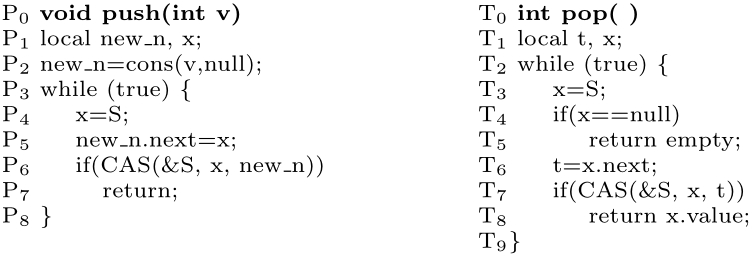}
   \caption{the Treiber stack  }
 \end{figure}
A push operation always inserts a node with a value into a list;
A pop operation  either
 removes a node from  a list and returns the value of the node or returns $empty$;
 A node is removed  at most once.
Thus, for a complete history $ H$ of the
Treiber stack,
 there is a safe mapping $\textit{Mat}$
 from $\textit Pop(H)$  to $\textit Push(H)\cup\{\epsilon\}$.

For a non-empty pop operation,  the final
$ T_7$  action   successfully removes the head node; for an empty pop operation,  when  the final $ T_3$  action is executed, the list is empty.
For any execution,   the linear order  of   pop operations is   constructed in terms of  these actions (i.e., the pop operations are arranged in the execution order of these removing actions).

For any execution,  the extension $\prec_{e}$ of  the happened-before order $\prec_{H}$
is   constructed in terms of the following actions:
for a push operation, we choose the final
$P_6$ action  (which successfully inserts  a new node);
for a non-empty pop operation,  we choose the final
$ T_7$ action; for an empty pop operation, we choose the final $ T_3$  action.
For two operations $op_1$ and $op_2$, if the chosen action of $op_1$  precedes the one of $op_2$ in the execution, then $op_1\prec_{e}op_2$. In a complete history,  the extension $\prec_{e}$ is  a linear order.
If a non-empty pop operation is smaller than  a push operation w.r.t. $\prec_{e}$,
the pop operation cannot observe the effect of the push operation.
This is because the  push operation inserts a node into the list by $P_6$ after the pop operation
performs the removing action $ T_7$.

\begin{theorem}
Every complete history $H$  of the Treiber stack is linearizable. \end{theorem}

\begin{proof}

Assume$ pop_1,pop_2,\dots,pop_n$ is the
linearization of pop operations in $H$  constructed by the above way.
Based on  the safe mapping  $\textit{Mat}$ and the extension $ \prec_{e}$  of $ \prec_{H}$,
we prove that $H$  satisfies the three conditions of Theorem \ref{stackTheorem}.

\vspace{0.5em}
1. $ \forall i,j.\; 1<j\leqslant n$,  $ 1\leqslant i<j. \;  \textit  Mat(pop_i)\neq \epsilon \Longrightarrow pop_j \nprec_H  \textit  Mat(pop_i)$.

\textbf{Proof.}  By the executing process,
the removing action $T_{7}$ of $ pop_i$ is executed after
the inserting node action $ P_6$ of $ \textit Mat(pop_i)$.
By the constructing way of the linear order of pop operations and  $ i<j$, the removing action of $ pop_j$ is executed after the removing action of $ pop_i$.
 Thus $ pop_j\nprec_H  \textit  Mat(pop_i)$.

\vspace{0.5em}
2.  $1\leqslant\forall i\leqslant n$. if $\textit  Mat(pop_i)\neq\epsilon$, let the set $B_i=\{push\mid push\in \textit  Push(H)-\{\textit Mat(pop_1),$ $\ldots,$ $\textit Mat(pop_{i-1})\} $ $ \wedge \; pop_i\nprec_{e}  push   \}$,   then  $ \textit Mat(pop_i)\in B_i $ and  $ \forall push $ $\in B_i.\;  \textit  Mat(pop_i)\nprec_{e} push$;

\textbf{Proof.}
By $ push\in B_i$, the node inserted by $ push$ is not removed
before the removing  action  of $ pop_i$.
Since $ pop_i\nprec_{e}  push$ (by $  push\in B_i$),
we get $ push\prec_{e} pop_i $ (because $ \prec_{e}$ is a linear order).
Since $ push\prec_{e} pop_i $, the inserting node action  of $push$ is executed  before the removing  action of $ pop_i$ (by the constructing rule of   $ \prec_{e}$).
Thus $ push$ has inserted a node into the list
before the removing  action of $ pop_i$.
The node inserted by $ \textit Mat(pop_i)$ is  a head node when $ pop_i$ removes the node.
In any state of the stack, the head node is the latest node.
Thus, we get $push \prec_{e}  \textit  Mat(pop_i)$.

\vspace{0.5em}
3. $1\leqslant\forall i\leqslant n$. if  $\textit  Mat(pop_i)=\epsilon$, let $A_i=\{\textit Mat(pop_1),\dots,\textit Mat(pop_{i-1})\}$, and $ C_i=\{push\mid push\simeq_H pop_i \wedge  push\in \textit  Push(H)-A_i \}$, then (a)
 $ \forall push\in \textit  Push(H).\; push\prec_H pop_i \Longrightarrow push \in A_i$; (b) $\forall push\in C_i ,\forall x\leqslant i-1. \; push \nprec_H pop_x $ and $ push\nprec_H \textit Mat(pop_x)$.

\textbf{Proof.} Since $ pop_i$ returns $ empty$, the list is empty at the time point when the statement $ T_ {3}$ is executed.
  Thus for any push operation $ push $ which precedes $ pop_i$, the value pushed by $ push $  is removed by a previous pop operation (i.e., $ pop_x,\;x<i$).
  Thus $ push \in A_i$.
Since at the above time point the list is empty,
 the push operations in the set $ C_i $ do not complete the inserting node actions (at $ P_6$). However, the  previous  pop operations and their matching push operations complete the removing and inserting node actions, respectively. Thus the second clause of the condition is satisfied.
\end{proof}

\subsection{Verifying the Time-Stamped Stack}
Proving linearizability of the TS stack is challenging, because neither of its push and pop methods have ``fixed linearization points'' (see \cite{TSStack}).
Fig. 6 shows the pseudo code for  the TS stack.
We use $\prec_{ts}$  operator for timestamp comparison.
There are a number of  implementations of the time stamping algorithm \cite{5haas}.
  All these implementations guarantee that (1) in a sequential execution of two calls to the algorithm, the latter returns a bigger timestamp than the former and (2) a concurrent and overlapping execution of two calls to the algorithm  generates two  incomparable  timestamps.

\begin{figure}[ht]
  \centering
  \includegraphics[height=0.6\textwidth]{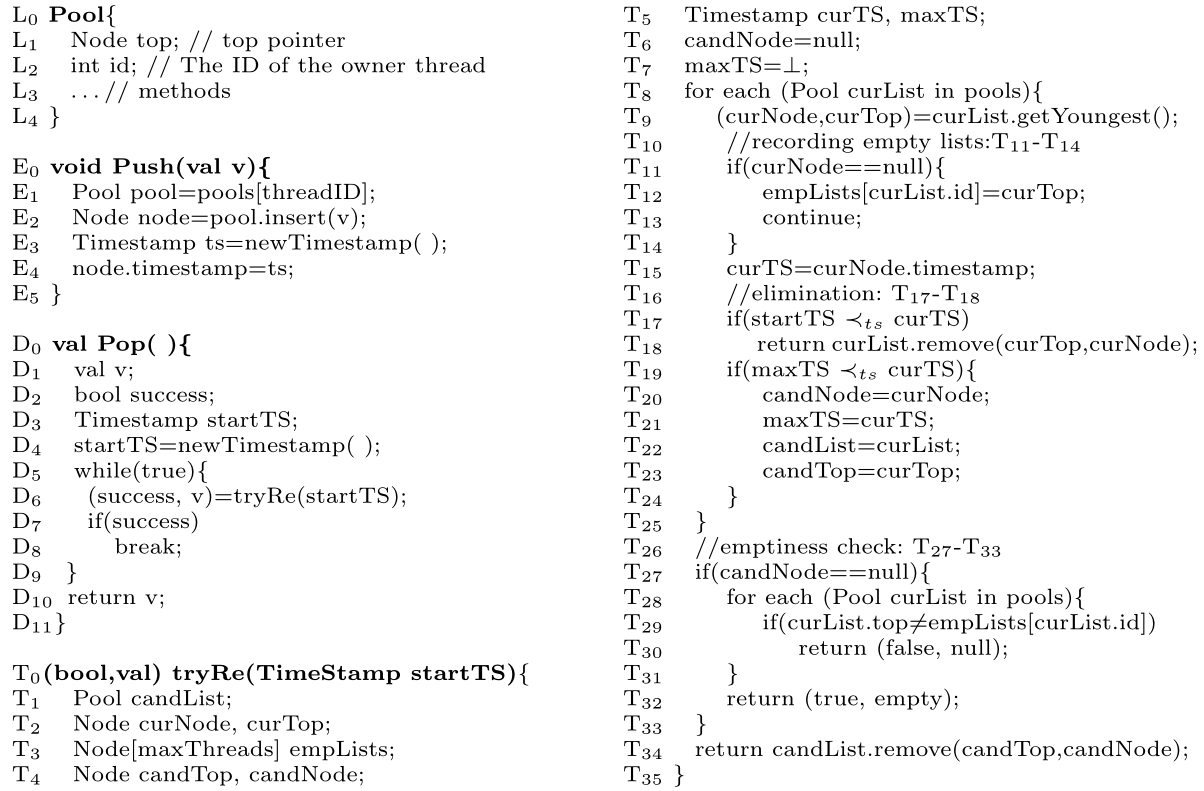}
  \caption{the TS stack }
 \end{figure}

$\textit Pool$ is a singly linked list with a top pointer $top$.
This stack maintains an
array \textit{pools} of singly-linked lists (i.e., instances of
$\textit Pool$), one for each thread (identified by $id$ of $\textit Pool$). A thread only inserts elements into its associated list by push operations.
Each node of $\textit Pool$ contains  a timestamp field $timestamp$, a data value field and a next pointer field.
 The methods of the singly-linked list $\textit Pool$  are as follows.  These methods  are linearizable, and can be viewed as atomic actions.
\begin{itemize}
\item $ insert(v)$ - inserts a node  with a value $v$ and a timestamp $\top$,  to the head of the list and returns a reference to the new node.
\item $ getYoungest()$ - returns a reference to the node with the youngest timestamp, or  $null$ if the list is $ empty$, together with the top pointer of the list.
\item $ remove(top, node)$ - tries to remove the given node from the list, where top denotes the  top pointer of the list. Returns $true$ and the value of the node if it succeeds, or returns \textit{false} and $null$ otherwise.
\end{itemize}

A $ push$ operation of a thread first  inserts a new node into  its associated list (line $ E_2$), then  generates a timestamp (line $ E_3$) and sets the timestamp field of the new node to the generated timestamp (line $ E_4$).

A $ pop$ operation first generates a timestamp \textit{startTS} (line $ D_4$),
attempts to remove an element  by calling the  $ tryRe$ method.
The   $ tryRe$ method traverses  all lists, searching for the node with the youngest timestamp to remove (line $ T_8-T_{25}$). If the youngest node has been found, the method tries to remove the node (line $ T_{34}$).

During the traversing of  the $ tryRe$ method, if it finds a node whose timestamp is bigger than the timestamp \textit{startTS} of the pop operation (line $ T_{17}$), then the $ tryRe$ method tries to remove the node (line $ T_{18}$).
In this case, the node is pushed during the execution of the current pop operation. Thus, the push operation which inserts the node is interleaved with the current
pop operation. If the $ tryRe$ method succeeds in removing the node,  we call them a timestamp elimination pair.

During the traversing of  the $ tryRe$ method, if it finds  an empty list, then the  top pointer of the list is recorded in the array \textit{empLists} (line $ T_{11}$-line $ T_{14}$).  After the traversing of  the $ tryRe$ method,  if no candidate node for removal is found, then the $ tryRe$ method checks whether the top pointers of all lists have changed (line $ T_{27}$-line $ T_{33}$).
If not, the $ tryRe$ method returns \textit{empty}, and then the pop operation returns \textit{empty}.
In this case, all lists must be empty when $ T_{27}$ begins to execute.
Otherwise, the $tryRe$  method returns \textit{false} (line $ T_{30}$),  and then the pop operation restarts.

A push operation always inserts a node with a value into a list;
A pop operation either  removes a node from  a list and returns the value of the node or returns $empty$; A node is removed  at most once.
Thus, for a complete history $ H$ of the
TS stack, there is a safe mapping $\textit Mat$  from $\textit Pop(H)$  to $\textit  Push(H)\cup\{\epsilon\}$.

In terms of Theorem \ref{elimiTheorem},  the timestamp elimination pairs  do not
 violate linearizability of concurrent stacks.
Thus, the timestamp  elimination pairs can be ignored  when we prove linearizability of the TS stack. The following  theorem and lemmas only considers  the histories without containing the timestamp  elimination pairs.

Theorem \ref{tsstackTheo}  states that the TS stack is linearizable, and the following two
lemmas are used in the proof of the theorem.

\begin{lemma} \label{tss1}
For any  pop operation $po$ and any
 push operation $pu$ in an execution, if  $po\nprec_{ts} pu$, then
  $pu$  inserts a node into a list  before $po$ begins to  traverse the lists.
\end{lemma}

\begin{proof}The timestamp of  a pop operation is generated by $ D_4$ action before it begins to traverse the lists. The timestamp of a push operation  is generated by $ E_3$ action after it inserts  a node into a list at $ E_2$.
Since $ po\nprec_{ts} pu$,   the $ E_3$ action of  $ pu$ precedes or is interleaved with  the $ D_4$ action of  $ po$. In either case,
 $ pu$ inserts a node into a list  before $ po$ begins to  traverse the lists.
\end{proof}

\begin{lemma} \label{tss2}
For any non-empty pop operation $po$ and for any
 push operation $pu$ in an execution, if the node inserted by  $pu$ is accessed by $po$ during its final traversal and $ \textit Mat(po)\neq pu$,  then  $  \textit  Mat(po) \nprec_{ts} pu$.
\end{lemma}

\begin{proof}
The timestamp of  the node inserted by  $  \textit  Mat(po)$ is  maximal among the nodes which  are accessed  by $ po$ during its final traversal (by $ T_{19}-T_{24}$).
A timestamp of a node is  generated by its corresponding push operation.
Thus,  $   \textit  Mat(po) \nprec_{ts} pu$.
\end{proof}

We choose $ \prec_{ts}$ as    the extension  of  $\prec_{H}$.
For two operations $\textit op_1$ and $\textit op_2$,  $ op_2\prec_{ts} op_1 $  if the timestamp generated by the operation $op_1$  is bigger  than  the one  generated by the operation $ op_2$.

\begin{theorem}\label{tsstackTheo}
For any complete history of the TS stack, let $H$ be a subsequence obtained by deleting  the timestamp elimination pairs of the history. $H$ is linearizable with respect to the standard   sequential   stack specification. \end{theorem}

\begin{proof}
For a non-empty pop operation, we choose the  successful removing action $ T_{34}$ to construct the linearization of pop operations; for an empty pop operation, we choose  the final $ T_{27}$  action (at the time point, all lists are empty).
Assume
$pop_1,\dots,pop_n$
is   the linearization  of  pop operations constructed in terms of  these atomic actions.
We choose $ \prec_{ts}$ as    the extension  of  $ \prec_{H}$.
Based on the linearization of pop operations, the safe mapping  \textit{Mat} and
the extension $\prec_{ts}$,
we prove that the TS stack satisfies the three conditions of Theorem \ref{stackTheorem}.

\vspace{0.5em}
1. $ \forall i,j.\; 1<j\leqslant n$,  $  1\leqslant i<j. \;  \textit  Mat(pop_i)\neq \epsilon \Longrightarrow pop_j \nprec_H  \textit  Mat(pop_i)$.

\textbf{Proof.}  By the executing process,
the inserting node action $ E_2$ of $  \textit  Mat(pop_i)$ precedes the removing  action  of $ pop_i$.
 By the constructing rule of the linearization of pop operations and $ i<j$, the removing  action of $ pop_i$ precedes
the one   of $ pop_j$.
 Thus $ pop_j\nprec_H  \textit  Mat(pop_i)$.

\vspace{0.5em}
2.
$1\leqslant\forall i\leqslant n$. if $\textit  Mat(pop_i)\neq\epsilon$, let the set $B_i=\{push\mid push\in \textit  Push(H)-\{\textit Mat(pop_1),$ $\ldots,$ $\textit Mat(pop_{i-1})\} $ $ \wedge \; pop_i\nprec_{ts}  push   \}$,   then  $\textit Mat(pop_i)\in B_i $ and  $ \forall push $ $\in B_i.\; \textit Mat(pop_i)\nprec_{ts} push$.

\textbf{Proof.}  Since $ H$ does not contain timestamp elimination pairs,
the timestamp $ startTS$ of $ pop_i$ is not smaller than the one of $  \textit  Mat(pop_i)$.
 Thus,  $\textit  Mat(pop_i)\in B_i$.
Since $ \forall  push\in B_i$,  $pop_i\nprec_{ts}  push$ (by the property of the set $B_i$).
Since $pop_i\nprec_{ts}  push$, the push operation $ push$ inserts a node into a list before $pop_i$ begins to traverse the lists (by Lemma \ref{tss1}).
Thus, by Lemma \ref{tss2}, $\textit  Mat(pop_i) \nprec_{ts} push$.

\vspace{0.5em}
3.
$1\leqslant\forall i\leqslant n$. if  $\textit  Mat(pop_i)=\epsilon$, let $A_i=\{\textit Mat(pop_1),\dots,\textit Mat(pop_{i-1})\}$, and $ C_i=\{push\mid push\simeq_H pop_i \wedge  push\in \textit  Push(H)-A_i \}$, then (a)
 $ \forall push\in \textit  Push(H).\; push\prec_H pop_i \Longrightarrow push \in A_i$; (b) $\forall push\in C_i ,\forall x\leqslant i-1. \; push \nprec_H pop_x $ and $ push\nprec_H \textit Mat(pop_x)$.

\textbf{Proof.} Since $ pop_i$ returns $ empty$,  all lists are empty at the time point when the final $ T_ {27}$ action of $ pop_i$ is executed.
  Thus for any push operation $push$ which precedes $ pop_i$, the value of $push$  is removed by a  pop operation ahead of $ pop_i$.
  Thus $ push \in A_i$. At the time point when the final $ T_ {27}$ action is executed,
 the push operations in the set $ C_i$ do not complete the inserting node actions (line $ E_2$). However, the  previous  pop operations and their matching push operations complete the removing and inserting node actions, respectively. Thus the second clause of the condition is satisfied.
\end{proof}

Consider the following  execution of the TS stack (adapted from \cite{18}),  depicted in  Fig. 7. We will construct a linearization of the execution  by the method  which is used in the proof of Theorem \ref{stackTheorem}.
  $ push(v,t)$/$ pop(v,t)$ denotes the push/pop operation which pushes/pops the value $v$ and generates the timestamp $t$. The black circles of the pop operations in Fig. 7  stand for the  removing  actions.

\begin{figure}[ht]
 \centering
  \includegraphics[height=0.27\textwidth]{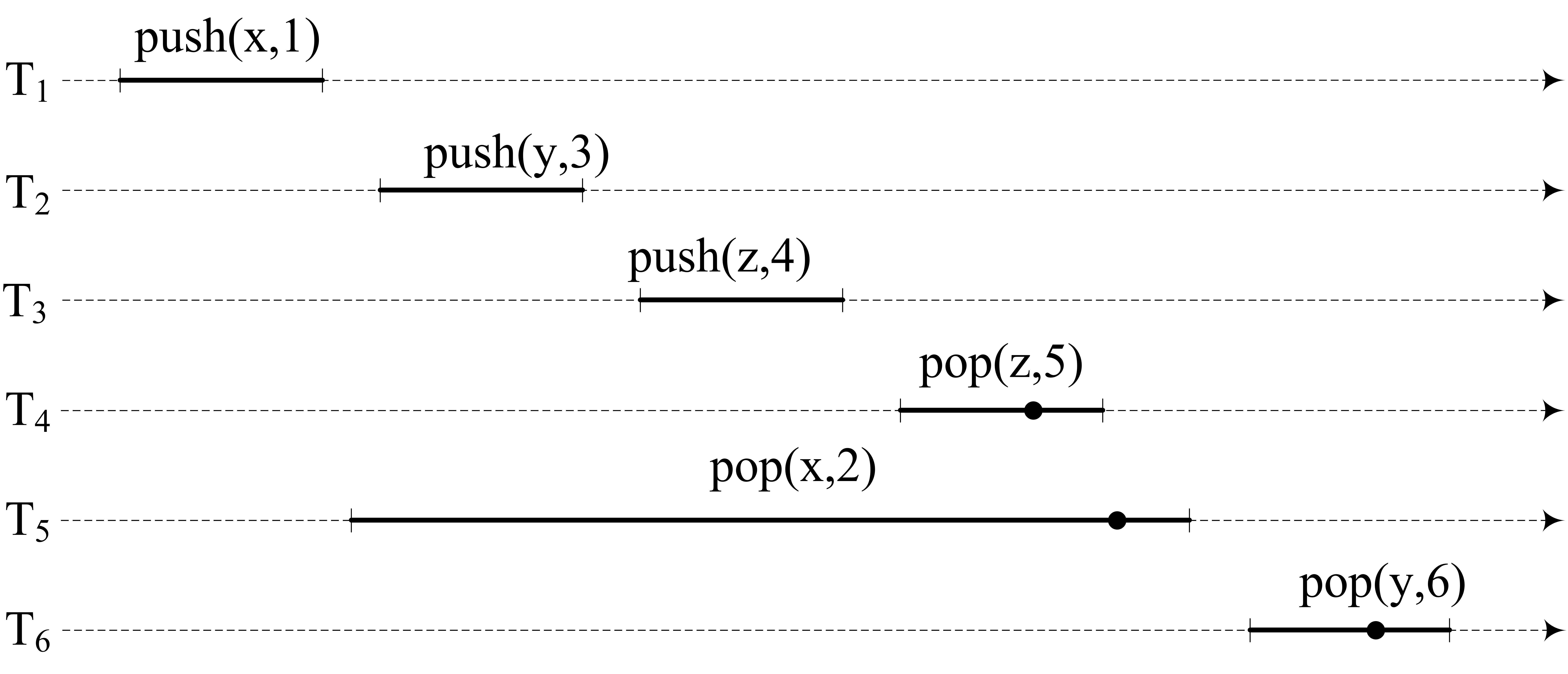}
  \caption{Example execution illustrating how to construct a linearization}
 \end{figure}

First, we construct the following  linearization of push operations by the partial order $ \prec_{ts}$.

\vspace{0.5em}
\begin{center} $ push(x,1),\;push(y,3),\;push(z,4)$ \end{center}
\vspace{0.5em}

Second, we construct the  initial  linearization of pop  operations in terms of their removing actions.

\vspace{0.5em}
\begin{center} $ pop(z,5),\;pop(x,2),\;pop(y,6)$ \end{center}
\vspace{0.5em}

Finally, we insert the pop operations, one after another, into the linearization sequence of push operations by using Algorithm 2.
Since there are no push operations whose timestamps are bigger than  $ pop(z,5)$, by  Algorithm 2, we insert $ pop(z,5)$ to the end of the push sequence and get the following sequence.

\vspace{0.5em}
\begin{center} $ push(x,1),\;push(y,3),\;push(z,4),\;pop(z,5)$ \end{center}
\vspace{0.5em}

By Algorithm 2 and $  pop(x,2)\prec_{ts} push(y,3) $, we  insert $ pop(x ,2)$ before $  push(y,3)$ and get the following sequence.

\vspace{0.5em}
\begin{center} $  push(x,1),\;pop(x,2),\;push(y,3),\;push(z,4),\;pop(z,5)$ \end{center}
\vspace{0.5em}

Since there are no push operations whose timestamps are bigger than $ pop(y,6)$,  we insert $ pop(y,6)$  to the end of the sequence and get the final sequence. Obviously, the  final sequence is a linearization of the execution.

\vspace{0.5em}
\begin{center} $ push(x,1),\;pop(x,2),\;push(y,3),\;push(z,4),\;pop(z,5),\;pop(y,6)$ \end{center}
\vspace{0.5em}

If the timestamp of the pop operation of Thread 5 is 3, for the pop operation, the corresponding set in the second condition is  $\{push(x,1), push(y,3) \}$.
By Lemma \ref{tss1}, $push(x,1)$, $ push(y,3)$ insert their values into the lists  before the pop operation begins to  traverse the lists. By Lemma \ref{tss2}, the pop operation of Thread 5 should pop the value pushed by $ push(y,3)$.

\section{Related Work and Conclusion }

There has been a great deal of work on linearizability verification
\cite{14survey,17,18,19,21,o22,o23,o24,o25,23,24,VerifyingA2,VerifyingA3,VerifyingA4,jia2}.
However, proving linearizability of  sophisticated concurrent data structures is still challenging.

Aspect-Oriented approach \cite{9h} is a property-based approach that is related to our work. The approach reduces
the problem of proving linearizability of concurrent queues to checking four simple properties, each of which can be
proved independently by simpler arguments.
For the non-empty dequeue operations, their  proof technique needs to verify the following  key property:  for two non-overlapping enqueue operations $ enq_1$ and $ enq_2$, if $ enq_1$ precedes $ enq_2$, then the value inserted by $ enq_2$ cannot be removed earlier than the one inserted by $ enq_1$.  {\"O}hman et al. \cite{ohman2022visibility} has extended   the technique   to snapshot data structures and  has axiomatized the important internal properties shared by several snapshot data structures.

However, they do not extend the approach to concurrent stacks.
Our work  is inspired by their work.
Our basic idea is that for any linearizable execution of a concurrent stack, every pop operation must pop the value pushed by a latest push operation of the
ones whose effects can be observed by the pop operation, and  we reduce the problem of proving linearizability of concurrent stacks to establishing a set of conditions based on the happened-before order of operations.

Bouajjani et al. \cite{21}  propose a forward simulation technique for proving linearizability.
They have successfully applied
the method to prove the TS stack and the HW queue. In fact, there does not exist a forward
simulation between the TS stack and the standard sequential stack. They need to construct a deterministic atomic reference implementation
(as an intermediate specification) for the TS stack, and the linearizability proof is reduced to showing that the TS stack
is forward-simulated by the intermediate specification.

Khyzha et al. \cite{18} propose a proof technique based on partial orders that is related to our work. The key idea
of their technique is to incrementally construct an abstract history—a partially ordered history of operations; the
linearizability proof of a concurrent data structure is reduced to establish a simulation between its execution and a
growing abstract history. They formalise the technique as a program logic based on rely-guarantee reasoning. Their proof technique is generic and can handle concurrent data structures with non-fixed linearization points. They have
applied the proof technique to prove the HW queue, the TS queue and the Optimistic set and conjecture that their
proof technique can be applied to prove the TS stack.

There are also lots of  automatic proof techniques  based on tools.  For example,
Ike  et al. \cite{mulder2023proof} develop proof automation  for linearizability in  separation logic.
They implement the proof automation in Coq by extending and generalizing Diaframe, a proof automation extension for Iris.
Abdulla et al. \cite{abdulla2018fragment} extend the observer automata to
the data structures based on heap structure by using a  fragment abstraction. However, for the
priority queues and sets, their  appraoch still requires users to annotate methods with linearization points.

\textbf{Conclusion} We have presented a novel proof technique for verifying linearizability of concurrent stacks. Our key contributions include (1) proving the soundness of the elimination mechanism, a common optimization in concurrent stacks, and (2) developing a stack theorem that simplifies linearizability proofs by reducing them to checking a set of conditions on the happened-before order. The main insight behind the stack theorem is to use an extended partial order to capture when a pop operation can observe the effect of a push operation.
We have demonstrated the effectiveness of our approach by applying it to verify two non-trivial concurrent stack algorithms: the Treiber stack and the Time-Stamped stack. The stack theorem greatly simplifies the linearizability proofs for these algorithms compared to prior approaches.
Our proof technique provides a systematic and compositional way to verify concurrent stacks with non-fixed linearization points, a challenging class of algorithms. In future work, we plan to investigate extending our approach to other concurrent data structures such as queues and sets. We also aim to integrate our technique with automated verification tools to enable scalable and push-button verification of concurrent data structures. Finally, we are interested in exploring how our approach can be adapted to reason about relaxed consistency, such as quasi-linearizability \cite{henzinger2013},   intermediate value linearizability \cite{rinberg2023} and interval-linearizability \cite{castaneda2018}.

\textbf{Funding}
This work is supported by National Natural Science Foundation of China [No.62341204] and  Science and Technology Research Project of Jiangxi Province Educational Department [No.GJJ2203609].

\appendix
\section{Algorithms}\label{A2}
Obviously, both of  the following  algorithms  can do the job.
 Algorithm 1  is  used in the proof of  Lemma \ref{elimiLemma}.  Algorithm 2  is  used in the proof of  Theorem \ref{stackTheorem}.

\begin{multicols}{2}
\begin{algorithm}[H]
 \SetAlgoNoLine
\caption{}
\uIf {$ L_n\prec L'$}
{ $L'$ is inserted after $L_n$\;
     \qquad \qquad\qquad$\dots$\;}
\uElseIf {$L_i\prec L'$}
{$L'$ is inserted between  $L_{i}$ and $L_{i+1}$\;
   \qquad \qquad\qquad$\dots$\; }
\uElseIf {$L_1\prec L'$}
{ $L'$ is inserted between $L_1$  and $L_2$\;}
\uElse {
 $L'$ is inserted before  $L_1$\;}
\end{algorithm}
\columnbreak
\begin{algorithm}[H]
 \SetAlgoNoLine
\caption{}
\uIf {$ L'\prec  L_1 $}
{ $L'$ is inserted before $L_1$\;
     \qquad \qquad\qquad$\dots$\;}
\uElseIf {$L'\prec L_i $}
{$L'$ is inserted between  $L_{i-1}$ and $L_{i}$\;
   \qquad \qquad\qquad$\dots$\; }
\uElseIf {$ L'\prec L_n $}
{ $L'$ is inserted between $L_{n-1}$  and $L_n$\;}
\uElse {
 $L'$ is inserted after  $L_n$\;}
\end{algorithm}
\end{multicols}

\bibliographystyle{elsarticle-num}
\biboptions{numbers,sort&compress}
\bibliography{sample-base}



\end{document}